\documentclass{article}

\usepackage{xspace}
\usepackage{tabularx}
\newcolumntype{L}{>{\raggedright\arraybackslash}X}
\usepackage[usenames,dvipsnames,table]{xcolor}
\usepackage{times}
\usepackage{soul}
\usepackage{url}
\usepackage[hidelinks]{hyperref}
\usepackage[utf8]{inputenc}
\usepackage[small]{caption}
\usepackage{graphicx}
\usepackage{amsmath,amsthm}
\usepackage{booktabs}
\usepackage{algorithm}
\usepackage{algpseudocode}
\usepackage[switch]{lineno}
\usepackage{tikz}
\usepackage{mathtools}
\usepackage{multirow}
\usepackage{cite} % sorting citation numbers
\usepackage{kotex}
\usepackage{pgfplots}
\pgfplotsset{compat=1.18}
\usepackage{tcolorbox}
\usepackage{enumitem}

\newtheorem{problem}{Problem}
\newtheorem{definition}{Definition}
\newtheorem{example}{Example}
\newtheorem{lemma}{Lemma}
\newtheorem{corollary}{Corollary}
\newtheorem{theorem}{Theorem}
\newcommand{\RecuRegex}{{\mdseries\scshape Re\-Syn}\xspace}
\newcommand{\resyn}{{\mdseries\scshape Re\-Syn}\xspace}
\newcommand{\setregex}{{\mdseries\scshape Set2Regex}\xspace}
\newcommand{\partitioner}{{\mdseries\scshape Partitioner}\xspace}
\newcommand{\segmenter}{{\mdseries\scshape Segmenter}\xspace}
\newcommand{\router}{{\mdseries\scshape Router}\xspace}
\newcommand{\prax}{{\mdseries\scshape Prax}\xspace}
\newcommand{\forest}{{\mdseries\scshape Forest}\xspace}
\newcommand{\splitregex}{{\mdseries\scshape Split\-Regex}\xspace}

\newcommand{\sregex}{\textit{Structured\-Regex}\xspace}
\newcommand{\snort}{\textit{Snort}\xspace}
\newcommand{\regexlib}{\textit{RegExLib}\xspace}
\newcommand{\polyglot}{\textit{Polyglot}\xspace}
\newcommand{\python}{\textit{Python}\xspace}

\newcommand{\hse}{Hierarchical Set Encoder\xspace}

\graphicspath{{images/}}

\definecolor{deltaBg}{RGB}{220,230,255}
\newcommand{\oraclerow}{\rowcolor{gray!30}}
\newcommand{\ourrow}{\rowcolor{deltaBg}}

\usepackage{ijcai26}

\title{\resyn: A Generalized Recursive Regular Expression Synthesis Framework}

\author{
Seongmin Kim$^1$\and
Hyunjoon Cheon$^2$\and
Su-Hyeon Kim$^{3}$\and
Yo-Sub Han$^3$\And
Sang-Ki Ko$^1$\\
\affiliations
$^1$University of Seoul\\
$^2$Gyeongsang National University\\
$^3$Yonsei University\\
\emails
\{mrseongminkim, sangkiko\}@uos.ac.kr,
hyunjooncheon@gnu.ac.kr,
\{suhyeon.kim, emmous\}@yonsei.ac.kr
}

\begin{document}
\maketitle

\begin{abstract}
Existing Programming-By-Example (PBE) systems often rely on simplified benchmarks that fail to capture the high structural complexity of real-world regexes, such as deeper nesting and frequent use of union operations.
To overcome the resulting performance drop, we propose \resyn, a synthesizer-agnostic divide-and-conquer framework that decomposes complex synthesis problem into manageable sub-problems.
We also introduce \setregex, a parameter-efficient synthesizer capturing the permutation invariance of examples.
Experimental results demonstrate that \resyn significantly boosts accuracy across various synthesizers, and its combination with \setregex establishes a new state-of-the-art on challenging real-world benchmark.
The complete source code, datasets, and pre-trained model checkpoints are publicly available at \url{https://github.com/mrseongminkim/ReSyn}.
\end{abstract}
\section{Introduction}
Regular expressions (regexes) are a fundamental tool in a wide range of domains, including text processing, data validation, and network security.
Despite their expressive power, regexes are notoriously difficult to write correctly due to their complex and unintuitive syntax, even for experienced practitioners.
To mitigate this difficulty, Programming-by-Example (PBE) approaches, which automatically synthesize regexes from input-output examples, have been studied extensively.

Recently, deep learning-based synthesizers leveraging large language models (LLMs) and encoder-decoder architectures have demonstrated promising results, achieving higher accuracy and significantly faster synthesis compared to traditional enumeration-based methods~\cite{VaduguruFP24}.
However, we argue that existing neural approaches suffer from a fundamental limitation: they are evaluated under overly simplified settings that fail to reflect the true structural complexity of real-world regexes.

\begin{table}[ht]
\centering
\setlength{\tabcolsep}{3.4pt}
\begin{tabular}{llccc}
\toprule
Dataset & Category &  Depth & Nodes & Alt \\
\midrule
\cite{VaduguruFP24} & Synthetic  & 3.39 & 6.45 & 0.04 \\
\midrule
\cite{ferreira2021forest} & Domain  & 3.03 & 8.02 & 0.05 \\
\midrule
\multirow{2}{*}{\snort\footnotemark} & Simplified  & 3.05 & 8.11 & 0.05 \\
                        & Real-World  & 3.13 & 9.60 & 0.06 \\
\midrule
\multirow{2}{*}{\textbf{\regexlib}\footnotemark[\value{footnote}]} & Simplified  & 3.31 & 11.04 & 0.48 \\
                                    & \textbf{Real-World}  & \textbf{4.40} & \textbf{23.69} & \textbf{1.25} \\
\bottomrule
\end{tabular}
\caption{Comparison of structural complexity.
To ensure consistency, all regexes were standardized using our Canonicalizer (\S\ref{subsec:canonicalization}).
Statistics are based on unique AST structures, abstracting concrete literals.
Simplified datasets are from \splitregex~\protect\cite{JALC-2025-157}.
Alt denotes the average number of Union operators.}
\label{tab:complexity}
\end{table}
\footnotetext{\url{https://github.com/lorisdanto/automatark}.}

Our analysis reveals a critical disconnect between existing benchmarks and practical regexes.
Prior works often rely on synthetic, domain-specific, or simplified datasets that restrict character classes and operators, artificially lowering synthesis difficulty.
As summarized in Table~\ref{tab:complexity}, genuine regexes from \regexlib are structurally deeper and semantically richer, featuring over $2\times$ more abstract syntax tree (AST) nodes than simplified versions and $3.6\times$ more than synthetic benchmarks.
Furthermore, while existing benchmarks are predominantly linear, real-world regexes heavily utilize the Union operator.
This discrepancy suggests that state-of-the-art models may overfit to simplified structures, failing to generalize to the recursive and nested nature of true practical instances.

This structural complexity poses severe challenges to existing paradigms. Enumeration-based solvers scale poorly with regex depth, while neural sequence-to-sequence (Seq2Seq) models suffer from two architectural inefficiencies: they flatten hierarchical ASTs into linear sequences, losing long-range dependencies, and they treat input examples as ordered sequences, violating the permutation invariance of sets. Although some divide-and-conquer strategies have been explored, they typically rely on rigid heuristics, such as assuming top-level Concatenation, and lack the true recursive capabilities needed for complex scenarios.

To overcome these limitations, we propose \resyn, a three-stage framework spanning data, model, and algorithm.
We first introduce a Regex Canonicalizer to standardize diverse regexes into a normal form for efficient training.
Building on this, we propose \setregex, a parameter-efficient (10M) base synthesizer equipped with a \hse that explicitly enforces permutation invariance.
Finally, we present a learning-based recursive framework comprising three specialized neural modules (\router, \partitioner, and \segmenter) that learn to adaptively decompose synthesis tasks.
This approach addresses the inherent complexity of optimal decomposition, which we prove to be NP-hard.
Our empirical results demonstrate that \resyn significantly improves synthesis success rates on complex benchmarks, effectively bridging the gap between simplified research settings and real-world requirements.

Our contributions are summarized as follows:
\begin{itemize}
    \item We quantitatively demonstrate the gap between existing benchmarks and real-world regexes, showing that they lack the structural complexity found in practice.
    \item We propose \setregex, a parameter-efficient model (10M) that matches the performance of a large-scale baseline (300M) by leveraging a \hse.
    \item We introduce \resyn, a recursive divide-and-conquer framework driven by learnable decomposition modules. Empirical results show that this framework significantly improves synthesis success rates on complex benchmarks like \regexlib.
    \item We prove that finding an optimal decomposition for regex synthesis is an NP-hard problem, highlighting the inherent computational intractability of the task and motivating the need for learned heuristics.
\end{itemize}

\section{Related Work}
Existing regex synthesis approaches can be broadly categorized into monolithic example-based methods and compositional divide-and-conquer strategies.
We position our framework within the latter, explicitly addressing the scalability and structural limitations inherent in the former.

\paragraph{Example-Based and Neural Regex Synthesis.}
Traditional example-based synthesis treats the task monolithically, inferring a complete regex in a single step via exhaustive or evolutionary search~\cite{BartoliDDMS14,LeeSO16}.
However, these methods struggle to scale with the combinatorial search space of complex patterns~\cite{OncinaG92,Gold78,LangPP98}.
Recent neural approaches frame synthesis as a sequence-to-sequence (Seq2Seq) problem, often leveraging frameworks like the Rational Speech Act (RSA) to improve accuracy~\cite{VaduguruFP24}.
Despite their inference speed, monolithic neural models face two critical hurdles.
First, relying solely on sequential token generation hinders capturing deep, nested structures.
Second, they typically ignore the permutation-invariant nature of input examples by flattening them into fixed sequences.
This artificial linearity forces the model to learn spurious orderings instead of underlying shared structural patterns.

\paragraph{Divide-and-Conquer Synthesis.}
Alur et al.~\cite{AlurRU17} proposed synthesizing programs by hierarchically partitioning examples, where the induced decision structure mirrors the program's control flow.
Farzan et al.~\cite{FarzanN21CAV,FarzanN21PLDI} further demonstrated that decomposing synthesis tasks enables parallel solving, significantly reducing overall synthesis time.
Similar principles appear in approaches that synthesize partial solutions over subsets of examples and compose them into a global solution~\cite{BarkePP20,ShrivastavaLT21}.

This paradigm has also been adapted for regex synthesis.
Raza et al.~\cite{raza2015compositional} decompose tasks using phrase-structure parsing of natural language descriptions, allowing sub-regexes to be synthesized independently.
Chen et al.~\cite{chen2023data} leverage LLMs to generate initial regex sketches, refining them by aligning components with substrings of positive examples.
While these approaches highlight the benefits of decomposition, they rely on auxiliary supervision—such as natural language descriptions or externally generated sketches—limiting their applicability in purely example-driven settings.

Recent works have attempted decomposition using only input-output examples. 
To handle structurally diverse patterns that inherently require Union operators, earlier evolutionary approaches~\cite{BartoliDMT16} bypassed explicit structural decomposition by employing a data-level separate-and-conquer strategy, iteratively extracting subsets of examples and joining the resulting sub-expressions with top-level Unions.
In contrast, recent decomposition-focused methods impose strong structural priors: \forest~\cite{ferreira2021forest} utilizes heuristic-based common substring identification, while \splitregex~\cite{JALC-2025-157} employs a neural model to predict segmentation points. 
However, a critical limitation of these modern approaches is that both rigidly assume that the top-level operator is always a Concatenation. 
This assumption severely hinders recursive decomposition. 
Because Concatenation is associative, repeatedly applying a Concatenation-only splitter results in a flat sequence of sub-problems rather than a deep hierarchical structure. 
Consequently, these methods lack the flexibility to handle nested structures where Union operators appear at the top level or are interleaved with Concatenations.

\paragraph{Positioning of Our Work.}
\resyn advances compositional synthesis by introducing a generalized recursive framework.
Unlike prior works that rely on fixed structural assumptions, we employ a learnable \router that dynamically determines whether to decompose the problem via Concatenation (using a \segmenter) or Union (using a \partitioner).
This flexibility allows for a principled, bottom-up composition that mirrors the complex syntactic structure of real-world regexes, achieved without requiring additional supervision.
\section{Proposed Method}
We present our recursive framework for example-based regex synthesis.
We begin by introducing the definitions and notation used throughout the paper, followed by a formal problem formulation.
We then describe our regex canonicalization pipeline, which standardizes diverse expressions into a unified representation.
Next, we present \setregex, a hierarchical encoder-decoder model that serves as the base synthesizer.
Building upon this foundation, we introduce the \resyn framework, a recursive divide-and-conquer approach that decomposes complex synthesis problems into simpler sub-problems.
Finally, we provide complexity analysis demonstrating that finding optimal decompositions is NP-hard, which motivates our learning-based approach.

\paragraph{Definitions and Notation.}
Throughout the paper, we use standard notation for alphabets, strings, and regular expressions.
An alphabet $\Sigma$ is a finite set of symbols, and $\Sigma^*$ denotes the set of all finite strings over $\Sigma$.
A regular expression (regex) $R$ describes a language $L(R) \subseteq \Sigma^*$.
The formal, recursive definition of regex syntax and semantics, including all operators and character classes, is provided in Appendix~\ref{sec:regex_def}.
We restrict our attention to regular constructs and exclude non-regular features such as backreferences.

We formally define the regex synthesis problem as follows:
Given a set of positive examples $P \subseteq \Sigma^*$ and negative examples $N \subseteq \Sigma^*$, 
the goal is to find a concise regex $R$ such that $P \subseteq L(R)$ and $N \cap L(R) = \emptyset$.
While a trivial solution exists by simply enumerating all positive examples with union ($w_1 | w_2 | \cdots | w_n$ for $P = \{w_1, \ldots, w_n\}$), such a regex lacks generalization and is impractically verbose. 
Therefore, we seek a regex that is both correct and concise, capturing the underlying pattern rather than merely memorizing examples~\cite{ijcai2024p717}.

\paragraph{Three-Stage End-to-End Neural Recursive Synthesis Framework: \resyn.}
\label{subsec:canonicalization}
We propose \resyn, a comprehensive framework designed to overcome the structural complexity of real-world regexes.
It is composed of three synergistic stages: data canonicalization, a synthesizer-agnostic recursive algorithm, and a specialized base model.

First, we address the issue of data quality. Real-world regex data exhibits significant syntactic diversity even when expressing identical languages, as different authors may write the same pattern in structurally different ways~\cite{DavisMCSL19}.
Such unnecessary syntactic variation introduces noise that hinders model learning and reduces data efficiency.
To address this challenge, we introduce a regex canonicalization pipeline that transforms diverse expressions into a standardized form.
Full details of the pipeline are provided in Appendix~\ref{sec:canonicalization}.

\label{sec:set2regex}
While our framework is compatible with any regex synthesizer, we propose \setregex, a parameter-efficient model designed to serve as a robust base synthesizer by capturing the permutation invariance of examples.
As illustrated in Figure~\ref{fig:set2regex}, the model utilizes a \hse that processes examples in two stages.
First, a character-level Transformer and Pooling by Multihead Attention (PMA)~\cite{lee2019set} compress each string into a fixed-size embedding $h_i$. 
These embeddings are processed by a string-level Transformer, which replaces traditional positional encodings with type embeddings indicating whether the string is a positive or negative example.
The resulting contextualized embeddings $\{h'_i\}$ are subsequently aggregated into a global context vector $c$.
The decoder generates the regex using a two-stage attention mechanism, attending first to the global context $c$ and then to the contextualized string embeddings $\{h'_i\}$.

\begin{figure}[ht!]
    \centering
    \includegraphics[width=\linewidth]{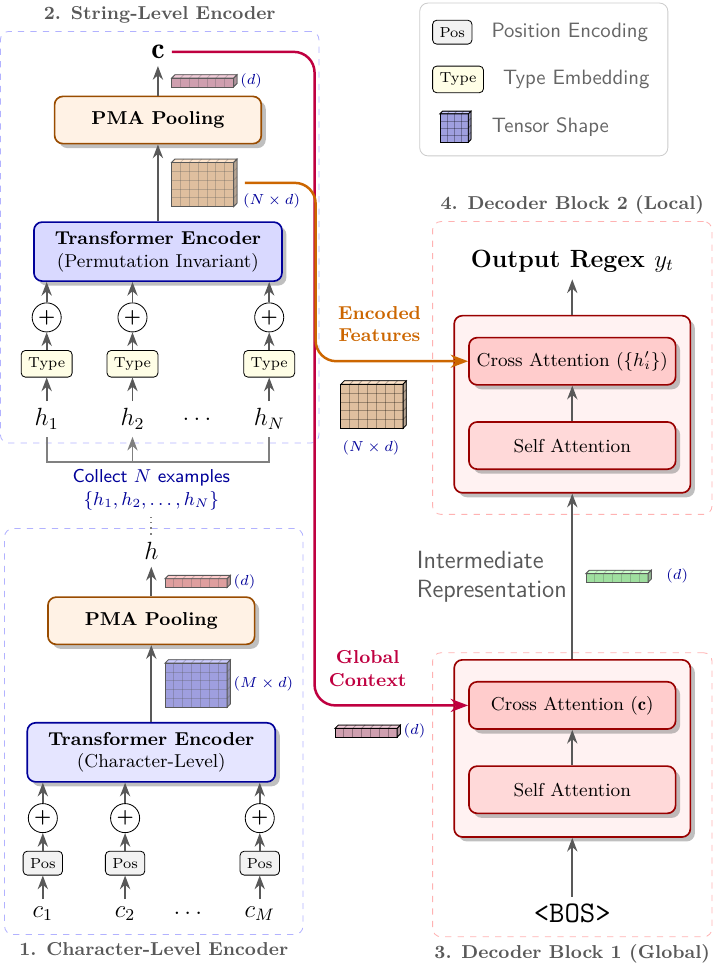}
    \captionof{figure}{The architecture of \setregex. 
    The \hse aggregates character features into string embeddings ($h_i$) and subsequently into a global context ($c$) via PMA to ensure permutation invariance. 
    The decoder employs a dual-attention mechanism, attending first to $c$ for global structure and then to $\{h'_i\}$ for local details.}
    \label{fig:set2regex}
\end{figure}

Building on the canonicalized data and the base synthesizer, we introduce the core of our framework: the recursive decomposition algorithm (see Appendix~\ref{app:algorithm} for the formal procedure and Figure~\ref{fig:flowchart} for the workflow).
Unlike prior works that rely on fixed heuristics, \resyn adaptively decomposes the problem using three specialized learnable modules: \router, \partitioner, and \segmenter.
This modular design allows \resyn to function as a wrapper that empowers any base synthesizer (including \setregex) to handle complex, nested patterns.

We define two complementary decomposition strategies utilized by the framework:
\begin{itemize}
    \item \textbf{Segmentation:}
    This strategy divides each string in the positive example set $P$ into $k$ segments based on shared logical positions (Concatenation).
    Each string $w \in P$ is split into segments $s_1, \ldots, s_k$, transforming the original set $P$ into an ordered sequence of $k$ sets $(P_1, \ldots, P_k)$.

    \item \textbf{Partitioning:}
    This strategy addresses structural diversity by grouping similar strings together (Union).
    The set $P$ is partitioned into $m$ disjoint subsets $\{P_1, \ldots, P_m\}$, where each subset is treated independently.
\end{itemize}

Figure~\ref{fig:running_example} (in Appendix) illustrates this recursive process.
By synthesizing partial regexes at the leaves---using the chosen base synthesizer---and composing them bottom-up, the system constructs the final solution in a principled manner.
Specific architectural details for the learnable decomposition modules are provided in Appendix~\ref{app:modules}.

\paragraph{Theoretical Hardness of Optimal Decomposition of Examples.}
To justify the use of neural architectures for regex synthesis, we establish the computational intractability of finding a concise regex.
We define the \textit{expression cost} $c_{E}(R)$ as the number of symbols (excluding operators) in regex $R$, and the \textit{language expression cost} $c_{E}(S)$ as the minimum $c_{E}(R)$ such that $L(R) \supseteq S$.

We first relate the expression cost to the string alignment problem.
An alignment of a set $S$ into $m$ tuples is a sequence $t_1 \dots t_m$ where each $t_j$ corresponds to a single character $\sigma \in \Sigma$ or $\lambda$.
The minimum alignment cost is denoted as $c(S)$.

\begin{lemma} \label{lem:cost-equiv}
For any finite language $S$, the language expression cost $c_{E}(S)$ is equal to the optimal alignment cost $c(S)$.
\end{lemma}

\begin{lemma} \label{lem:align-nphard}
The optimal alignment problem is \textit{NP-complete}.
\end{lemma}
\begin{proof}[Proof (Sketch)]
We reduce the \textit{Shortest Common Supersequence} (SCS) problem to the optimal alignment problem.
Given strings $\{w_1, \dots, w_n\}$, a common supersequence $s$ of length $m$ directly corresponds to an alignment of cost $m$.
Since SCS is NP-complete~\cite{RaihaU1981}, determining if $c(S) \le r$ is also NP-complete.
\end{proof}

Based on Lemma \ref{lem:cost-equiv} and \ref{lem:align-nphard}, we obtain the following result.
\begin{theorem}
\label{thm:concise-regex}
The Concise Regex Problem, which decides whether $c_{E}(S) \le r$ for a given set $S$ and integer $r$, is \textit{NP-complete}.
\end{theorem}

Note that the technical details, including the proof of Theorem \ref{thm:concise-regex} are provided in Appendix \ref{appendix:full_complexity}.

The NP-completeness implies that deterministic algorithms for optimal regex decomposition suffer from exponential time complexity as the number or length of examples increases.
In practice, finding the global optimum through combinatorial search is often intractable. 

This motivates the transition from an algorithmic approach to a learning-based approach.
The proposed approach, \resyn, can learn to approximate the optimal alignment and decomposition patterns.
By leveraging the generalizable representations of strings, neural architectures can efficiently navigate the vast search space of potential sub-regexes, providing near-optimal solutions in polynomial time where traditional symbolic methods fail to scale.

\section{Experiments}
\subsection{Experimental Setup}
\paragraph{Data Collection and Preprocessing.}
To train our neural models, we constructed a large-scale dataset of regular expressions by aggregating multiple sources: \polyglot~\cite{DavisMCSL19}, \python~\cite{chapman2016exploring}, and the \sregex~\cite{ye-etal-2020-benchmarking} dataset used in \prax~\cite{VaduguruFP24}.
Among these, \polyglot and \python are real-world regexes collected from open-source repositories, while \sregex is a synthetically generated dataset.
For evaluation, we employed three benchmarks: \regexlib and \snort, which contain real-world regexes, and the test split of \sregex, which is synthetic.

The detailed process of dataset construction, including canonicalization, structural filtering, sub-regex extraction, example generation (positive/negative), and substring expansion, is described in Appendix~\ref{app:data_details}.
We highlight the structural characteristics of our evaluation benchmarks in Figure~\ref{fig:combined_stats}.
Our evaluation suite consists of three distinct datasets: \sregex (334 instances), \snort (352 instances), and \regexlib (1,752 instances).
As illustrated in Figure~\ref{fig:combined_stats}, unlike the other baselines which are dominated by shallow Concatenations, \regexlib exhibits a diverse range of operators and significantly deeper AST structures, posing a greater challenge for synthesis.

\begin{figure}[t]
\centering
\begin{tikzpicture}[font=\sffamily]
    \definecolor{opConcat}{RGB}{0, 114, 178}
    \definecolor{opUnion}{RGB}{213, 94, 0}
    \definecolor{opRep}{RGB}{86, 180, 233}
    \definecolor{opClass}{RGB}{240, 228, 66}
    \definecolor{depthOne}{RGB}{254, 224, 210}
    \definecolor{depthTwo}{RGB}{252, 187, 161}
    \definecolor{depthThree}{RGB}{252, 146, 114}
    \definecolor{depthFour}{RGB}{251, 106, 74}
    \definecolor{depthFive}{RGB}{222, 45, 38}
    \definecolor{depthSix}{RGB}{80, 10, 10}
    \tikzstyle{sector}=[draw=white, line width=0.5pt]
    \tikzstyle{barRect}=[draw=white, line width=0.5pt]
    \matrix [draw=black!20, fill=white, rounded corners, anchor=south, column sep=0.2em] at (2.5, 1.4) {
        \node [shape=circle, fill=opConcat, scale=0.8, label=right:\small Concat] {}; &
        \node [shape=circle, fill=opUnion, scale=0.8, label=right:\small Union] {}; &
        \node [shape=circle, fill=opRep, scale=0.8, label=right:\small Repetition] {}; &
        \node [shape=circle, fill=opClass, scale=0.8, label=right:\small Char. Class] {}; \\
    };
    \begin{scope}[shift={(0,0)}]
        \node[align=center, font=\bfseries, below=1.2cm] at (0,0) {\sregex};
        \fill[opConcat, sector] (0,0) -- (90:1.0) arc (90:-266.76:1.0) -- cycle;
        \fill[opRep, sector] (0,0) -- (-266.76:1.0) arc (-266.76:-268.92:1.0) -- cycle;
        \fill[opUnion, sector] (0,0) -- (-268.92:1.0) arc (-268.92:-270:1.0) -- cycle;
    \end{scope}
    \begin{scope}[shift={(2.5,0)}]
        \node[align=center, font=\bfseries, below=1.2cm] at (0,0) {\snort};
        \fill[opConcat, sector] (0,0) -- (90:1.0) arc (90:-262.84:1.0) -- cycle;
        \fill[opRep, sector] (0,0) -- (-262.84:1.0) arc (-262.84:-270:1.0) -- cycle;
    \end{scope}
    \begin{scope}[shift={(5,0)}]
        \node[align=center, font=\bfseries, below=1.2cm] at (0,0) {\regexlib};
        \fill[opConcat, sector] (0,0) -- (90:1.0) arc (90:-179.39:1.0) -- cycle;
        \fill[opUnion, sector] (0,0) -- (-179.39:1.0) arc (-179.39:-240.41:1.0) -- cycle;
        \fill[opRep, sector] (0,0) -- (-240.41:1.0) arc (-240.41:-263.63:1.0) -- cycle;
        \fill[opClass, sector] (0,0) -- (-263.63:1.0) arc (-263.63:-270:1.0) -- cycle;
    \end{scope}
    \begin{scope}[yshift=-5.0cm]
        \foreach \y/\label in {0/0, 1.5/50, 3/100} {
            \draw[gray!30] (-1.0, \y) -- (6.0, \y);
            \node[left, font=\small, gray] at (-1.0, \y) {\label\%};
        }
        \node[rotate=90, font=\small\bfseries, gray] at (-2.1, 1.5) {Instance Ratio (\%)};
        \def\bw{1.0} 
        \begin{scope}[shift={(0,0)}]
            \fill[depthTwo, barRect]   (-\bw/2, 0) rectangle (\bw/2, 0.02);
            \fill[depthThree, barRect] (-\bw/2, 0.02) rectangle (\bw/2, 2.02);
            \fill[depthFour, barRect]  (-\bw/2, 2.02) rectangle (\bw/2, 2.81);
            \fill[depthFive, barRect]  (-\bw/2, 2.81) rectangle (\bw/2, 3.00);
            \node[below, font=\bfseries] at (0,0) {\sregex};
        \end{scope}
        \begin{scope}[shift={(2.5,0)}]
            \fill[depthTwo, barRect]   (-\bw/2, 0) rectangle (\bw/2, 0.18);
            \fill[depthThree, barRect] (-\bw/2, 0.18) rectangle (\bw/2, 2.62);
            \fill[depthFour, barRect]  (-\bw/2, 2.62) rectangle (\bw/2, 2.76);
            \fill[depthFive, barRect]  (-\bw/2, 2.76) rectangle (\bw/2, 2.98);
            \fill[depthSix, barRect]   (-\bw/2, 2.98) rectangle (\bw/2, 3.00);
            \node[below, font=\bfseries] at (0,0) {\snort};
        \end{scope}
        \begin{scope}[shift={(5,0)}]
            \fill[depthOne, barRect]   (-\bw/2, 0) rectangle (\bw/2, 0.05);
            \fill[depthTwo, barRect]   (-\bw/2, 0.05) rectangle (\bw/2, 0.27);
            \fill[depthThree, barRect] (-\bw/2, 0.27) rectangle (\bw/2, 1.40);
            \fill[depthFour, barRect]  (-\bw/2, 1.40) rectangle (\bw/2, 2.03);
            \fill[depthFive, barRect]  (-\bw/2, 2.03) rectangle (\bw/2, 2.62);
            \fill[depthSix, barRect]   (-\bw/2, 2.62) rectangle (\bw/2, 3.00);
            \node[below, font=\bfseries] at (0,0) {\regexlib};
        \end{scope}
        \matrix [draw=black!20, fill=white, rounded corners, anchor=north, column sep=0.2em] at (2.5, -0.8) {
            \node [shape=circle, fill=depthOne, scale=0.8, label=right:\normalsize D1] {}; &
            \node [shape=circle, fill=depthTwo, scale=0.8, label=right:\normalsize D2] {}; &
            \node [shape=circle, fill=depthThree, scale=0.8, label=right:\normalsize D3] {}; &
            \node [shape=circle, fill=depthFour, scale=0.8, label=right:\normalsize D4] {}; &
            \node [shape=circle, fill=depthFive, scale=0.8, label=right:\normalsize D5] {}; &
            \node [shape=circle, fill=depthSix, scale=0.8, label=right:\normalsize D6+] {}; \\
        };
    \end{scope}
\end{tikzpicture}
\caption{Structural Complexity Comparison across Benchmarks.
The top row displays the distribution of Top-level Operators; while \sregex and \snort are dominated by Concatenations, \regexlib{} exhibits a diverse structural composition with significant use of Union operators. 
The bottom row illustrates the AST Depth distribution using a sequential color gradient (darker indicates deeper nesting). 
Notably, \regexlib contains a significant proportion of high-complexity instances (Depth $\ge$ 5), whereas the other benchmarks are concentrated in shallower regions.}
\label{fig:combined_stats}
\end{figure}

\paragraph{Baselines.}
Since traditional enumerative synthesizers scale poorly to the complex, nested regular expressions targeted in this work, we focus our comparison exclusively on neural synthesis approaches.
We compare our framework against two categories of baselines: neural synthesizer models and decomposition-based approaches.

We evaluate against the following state-of-the-art neural synthesis models:
\begin{itemize}
    \item \prax: A ByT5-based Seq2Seq model~\cite{VaduguruFP24}.
    For a fair comparison, we fine-tuned \texttt{google/byt5-small} on our dataset for 10 epochs using BF16 precision and an effective batch size of 64 via gradient accumulation.
    
    \item \texttt{gpt-oss-120b} \& \texttt{GPT-5}: Large-scale reasoning models evaluated via identical zero-shot prompting (Appendix~\ref{app:prompt_details}).
    Due to its immense capacity, \texttt{GPT-5} is benchmarked separately against our beam search model to highlight our framework's efficiency rather than in the main comparison.
\end{itemize}

We compare our recursive decomposition strategy against approaches that employ problem splitting to validate the effectiveness of the strategy.
For fair comparison, we adopted only their decomposition strategies and combined them with neural synthesizer backend.
\begin{itemize}
    \item \splitregex~\cite{JALC-2025-157}: A neural strategy that segments positive strings. Since \splitregex is functionally identical to our \segmenter, we only include it as a baseline in the ablation study.
    \item \forest~\cite{ferreira2021forest}: A heuristic approach that identifies common substrings as separators to split the synthesis task.
    \item Oracle: Oracle variant of \resyn, performing recursive decomposition guided by the ground truth regex structure.
    It follows the same inference pipeline and termination logic as \resyn but utilizes ground truth labels for routing and decomposition decisions to establish an upper bound on performance.
\end{itemize}

\paragraph{Implementation Details.}
Our framework is highly parameter-efficient, totaling only 29.6M parameters (\setregex: 10M, \router: 4.6M, \partitioner: 7.5M, \segmenter: 7.5M).
We utilized greedy decoding by default for inference.
Detailed hyperparameters, architecture sizes, and training data coverage are provided in Appendix~\ref{app:training_details}.

\subsubsection{Evaluation Metrics}
Evaluating synthesized regular expressions in a PBE setting presents unique challenges.
A trivial solution that simply enumerates all positive examples using the Union operator satisfies the input constraints but fails to generalize and merely memorizes the data without capturing the underlying structure.
To comprehensively assess the quality of the synthesized regexes, we employ three complementary metrics that measure structural conciseness, strict functional equivalence, and approximate functional similarity.

We define the \textbf{Synthesis Success Rate} as the percentage of instances where the synthesized regex $R_{pred}$ correctly classifies all training examples. 
Additionally, to penalize trivial enumerations, we measure the \textbf{Conciseness Ratio}, defined as $|R_{pred}| / |R_{gt}|$, where $R_{gt}$ denotes the ground truth regex. 
Ratios closer to 1.0 indicate successful pattern generalization rather than mere data memorization.
To evaluate strict generalization, we use \textbf{Semantic Accuracy}, defined as the percentage of cases where the synthesized regex correctly classifies \emph{all} examples in the held-out set.
This is the strictest criterion, requiring the synthesized regex to be functionally equivalent to the ground truth within the scope of the test data.
However, because Semantic Accuracy can be overly punitive due to the under-specification problem (e.g., distinguishing \texttt{[a-z]*} from \texttt{[a-z]+} without empty string examples), we also measure the \textbf{Matthews Correlation Coefficient (MCC)} on the held-out set:
\begin{equation}
    \text{MCC} = \frac{\text{TP} \cdot \text{TN} - \text{FP} \cdot \text{FN}}{\sqrt{(\text{TP}+\text{FP})(\text{TP}+\text{FN})(\text{TN}+\text{FP})(\text{TN}+\text{FN})}}
\end{equation}
This provides a balanced correlation metric ranging from -1 to +1.
Finally, regarding \textbf{Failure Handling}, failed syntheses are recorded as failures for the Success Rate. 
For secondary metrics (Conciseness and MCC), failed instances are treated as a trivial union of positive examples, naturally penalizing them with excessive length and poor generalization scores.

\subsection{Experimental Results and Analysis}
We evaluate the proposed methods by addressing the following four research questions (RQs):
\begin{itemize}
    \item \textbf{RQ1:} Can the \resyn framework universally improve the performance of neural synthesizers?
    \item \textbf{RQ2:} Is recursive decomposition more effective than non-recursive (top-level only) splitting for complex real-world regexes?
    \item \textbf{RQ3:} Does the hierarchical architecture of \setregex achieve parameter efficiency while maintaining performance comparable to existing neural baselines?
    \item \textbf{RQ4:} Can our specialized framework compete with large-scale general-purpose reasoning models?
\end{itemize}

\begin{table*}[t]
\centering
\setlength{\tabcolsep}{4.8pt}
\begin{tabular}{lcccccccccccc}
\toprule
& \multicolumn{4}{c}{\textbf{\sregex}} & \multicolumn{4}{c}{\textbf{\snort}} & \multicolumn{4}{c}{\textbf{\regexlib}} \\
\cmidrule(lr){2-5} \cmidrule(lr){6-9} \cmidrule(lr){10-13}
\textbf{Method} & \textbf{Succ} & \textbf{Conc} & \textbf{Acc} & \textbf{MCC} & \textbf{Succ} & \textbf{Conc} & \textbf{Acc} & \textbf{MCC} & \textbf{Succ} & \textbf{Conc} & \textbf{Acc} & \textbf{MCC} \\
\midrule
\prax & 85.03 & 2.16 & 76.05 & 83.47 & 67.05 & 6.21 & 56.53 & 64.33 & 42.24 & 5.27 & 33.68 & 39.71 \\
~~\prax + Forest & 80.84 & 2.35 & 67.96 & 78.25 & 77.84 & 5.01 & \textbf{63.07} & \underline{73.01} & 49.94 & 4.76 & 36.76 & 46.24 \\
\ourrow~~\prax + \resyn & \underline{96.71} & \textbf{1.35} & \underline{84.43} & \underline{94.29} & \textbf{80.68} & \textbf{4.64} & \underline{61.93} & \textbf{74.65} & \underline{67.29} & \underline{4.00} & \underline{40.81} & 58.75 \\
\midrule
\setregex & 90.12 & 1.75 & 82.34 & 88.61 & 55.40 & 7.92 & 44.03 & 51.29 & 38.93 & 5.72 & 30.99 & 36.08 \\
~~\setregex + Forest & 85.03 & 2.02 & 73.95 & 82.72 & 76.70 & 5.31 & \underline{61.93} & 71.41 & 49.03 & 4.99 & 35.62 & 44.92 \\
\ourrow~~\setregex + \resyn & \textbf{97.60} & \underline{1.38} & \textbf{85.93} & \textbf{95.23} & \underline{79.55} & \underline{4.91} & 58.24 & 72.41 & \textbf{68.26} & \textbf{3.88} & \textbf{41.61} & \underline{58.97} \\
\midrule
\texttt{gpt-oss-120b} & 79.94 & 2.70 & 49.10 & 71.81 & 43.18 & 8.33 & 29.26 & 39.36 & \underline{67.29} & 4.03 & 40.75 & \textbf{59.06} \\
\midrule
\oraclerow \prax + Oracle & 99.10 & 1.14 & 87.43 & 96.78 & 86.36 & 3.62 & 65.91 & 79.64 & 80.59 & 3.29 & 41.72 & 65.40 \\
\oraclerow \setregex + Oracle & 99.10 & 1.21 & 88.62 & 96.99 & 86.36 & 3.57 & 63.92 & 78.55 & 82.53 & 2.85 & 43.15 & 66.51 \\
\bottomrule
\end{tabular}
\caption{Comprehensive comparison of synthesis methods across benchmarks.
Succ: Synthesis Success Rate (\%),
Conc: Conciseness Ratio (lower is better),
Acc: Semantic Accuracy (\%),
MCC: Matthews Correlation Coefficient ($\times 100$).
Best results in bold, second best underlined (excluding Oracle).
Highlighted rows indicate our proposed adaptive decomposition (\resyn) and the Oracle upper bound.
}
\label{tab:main_results}
\end{table*}

Table~\ref{tab:main_results} summarizes the quantitative results of our experiments across three benchmarks of increasing structural complexity: \sregex, \snort, and \regexlib.
To rigorously evaluate the impact of decomposition, we employ two distinct base synthesizers: the existing \prax model and our proposed \setregex.
For each base synthesizer, we compare our \resyn framework against the decomposition-based baseline, \forest, to isolate the efficacy of our recursive approach.

Overall, the results demonstrate that applying \resyn consistently improves performance across base models, outperforming the \forest baseline in most metrics.
Notably, the \setregex + \resyn configuration achieves state-of-the-art results, demonstrating performance comparable to or even surpassing the massive general-purpose reasoning model on the most complex benchmark across the majority of evaluation metrics, despite having significantly fewer parameters.
% We now examine each research question in detail.

\subsubsection{RQ1: Effectiveness of ReSyn Framework}
\resyn significantly enhances neural synthesizers across all benchmarks (Table~\ref{tab:main_results}).
For \prax, it boosts the \regexlib success rate from 42.24\% to 67.29\%.
The impact is greater for our \setregex model, where the \setregex + \resyn configuration achieves a 68.26\% success rate (+29.33\% absolute increase) and improves Semantic Accuracy to 41.61\%.
These gains are most pronounced in the complex real-world dataset \regexlib, where \resyn's recursive decomposition effectively mitigates structural complexity.
Furthermore, the lower Conciseness Ratio indicates that \resyn encourages compact, generalized patterns rather than trivial data memorization.
\begin{figure}[t]
\centering
\begin{tikzpicture}[scale=0.9]
\begin{axis}[
    width=\linewidth,
    height=6cm,
    xlabel={Regex AST Depth},
    ylabel={Synthesis Success Rate (\%)},
    xmin=1, xmax=6,
    ymin=0, ymax=100,
    xtick={1,2,3,4,5,6},
    xticklabels={1, 2, 3, 4, 5, 6+},
    ytick={0,20,40,60,80,100},
    legend pos=south west,
    legend style={font=\small, fill=white, fill opacity=0.9, draw opacity=1, text opacity=1},
    grid=major,
    grid style={dashed,gray!30},
    line width=1.2pt,
    mark size=2.5pt,
    every axis plot/.append style={thick}
]
\addplot[color=red, mark=square*] coordinates {(1, 87.1) (2, 73.68) (3, 56.73) (4, 44.4) (5, 31.71) (6, 19.2)};
\addlegendentry{\setregex}
\addplot[color=orange, mark=triangle*] coordinates {(1, 87.1) (2, 78.95) (3, 65.72) (4, 53.07) (5, 44.76) (6, 32.59)};
\addlegendentry{+ \forest}
\addplot[color=yellow!80!black, mark=diamond*] coordinates {(1, 87.1) (2, 72.37) (3, 80.46) (4, 69.77) (5, 63.17) (6, 43.3)};
\addlegendentry{+ \splitregex}
\addplot[color=blue, mark=pentagon*, dashed] coordinates {(1, 87.1) (2, 78.29) (3, 79.95) (4, 73.15) (5, 67.01) (6, 51.34)};
\addlegendentry{+ \resyn}
\addplot[color=green!60!black, mark=*] coordinates {(1, 87.1) (2, 80.26) (3, 89.2) (4, 92.18) (5, 79.8) (6, 63.84)};
\addlegendentry{+ Oracle}
\end{axis}
\end{tikzpicture}
\caption{Performance comparison across regex AST depth.
Non-recursive methods (\forest, \splitregex) show sharp degradation with increasing depth, while Recursive maintains robust performance.
The gap widens beyond depth 4, highlighting the necessity of recursive decomposition for complex real-world regexes.}
\label{fig:recursion_depth}
\end{figure}

\subsubsection{RQ2: Importance of Recursive Decomposition}
We compared \resyn against single-level decomposition baselines (\splitregex, \forest).
As shown in Figure~\ref{fig:recursion_depth}, non-recursive methods exhibit sharp performance degradation as regex AST depth increases.
While shallow splitting suffices for simple patterns (Depth $\le$ 3), it fails to address the combinatorial complexity of nested structures.
In contrast, \resyn maintains robust performance even at Depth 5 and 6+, validating the necessity of a recursive strategy.

\subsubsection{RQ3: Parameter Efficiency of Hierarchical Architecture}
\setregex (10M) is $30\times$ smaller than the \prax baseline (300M) but achieves comparable or superior performance.
This efficiency stems from our \hse, which eliminates the need to learn spurious example orderings, a common overhead in traditional Seq2Seq models.
Notably, \setregex outperforms \prax on the synthetic benchmark (90.12\% vs. 85.03\%) and matches it on \regexlib, proving that architectural alignment with set-theoretic inputs can decouple model capacity from synthesis accuracy.

\subsubsection{RQ4: Comparison with Advanced Language Models}
Compared to \texttt{gpt-oss-120b}, our 29.6M-parameter framework achieves higher Synthesis Success Rates and Semantic Accuracy across all benchmarks despite being $4,000\times$ smaller and avoiding expensive reasoning tokens.
We further evaluated \resyn (with $k=500$ beam search) against \texttt{GPT-5}.
As shown in Table~\ref{tab:gpt5_comparison}, \resyn outperforms \texttt{GPT-5} on \sregex and \snort, and notably achieves higher Semantic Accuracy on \regexlib.
This indicates that our recursive inductive bias captures ground-truth structures more effectively than general-purpose large-scale reasoning, ensuring superior generalization with a fraction of the computational footprint.

\begin{table*}[t]
\centering
\setlength{\tabcolsep}{5pt}
\begin{tabular}{lcccccccccccc}
\toprule
& \multicolumn{4}{c}{\textbf{\sregex}} & \multicolumn{4}{c}{\textbf{\snort}} & \multicolumn{4}{c}{\textbf{\regexlib}} \\
\cmidrule(lr){2-5} \cmidrule(lr){6-9} \cmidrule(lr){10-13}
\textbf{Method} & \textbf{Succ} & \textbf{Conc} & \textbf{Acc} & \textbf{MCC} & \textbf{Succ} & \textbf{Conc} & \textbf{Acc} & \textbf{MCC} & \textbf{Succ} & \textbf{Conc} & \textbf{Acc} & \textbf{MCC} \\
\midrule
\multicolumn{13}{l}{\textit{Large-scale Reasoning Model} (N/A Params)} \\
GPT-5 & 96.71 & 1.43 & 60.78 & 87.59 & 85.23 & 3.70 & 54.83 & 73.55 & \textbf{90.07} & \textbf{2.12} & 48.86 & \textbf{75.89} \\
\midrule
\multicolumn{13}{l}{\textit{Specialized Neural Synthesis} (\setregex + \resyn, 29.6M Params)} \\
Beam Search ($k=500$) & \textbf{100.00} & \textbf{1.02} & \textbf{89.52} & \textbf{97.79} & \textbf{90.62} & \textbf{3.28} & \textbf{63.64} & \textbf{82.10} & 85.33 & 2.58 & \textbf{50.29} & 73.65 \\
\bottomrule
\end{tabular}
\caption{Comparison between \texttt{GPT-5} and \resyn using beam search decoding ($k=500$). 
While \texttt{GPT-5}'s exact parameter count is undisclosed, our highly parameter-efficient \resyn framework significantly narrows the performance gap. Notably, beam search enables \resyn to achieve competitive or even superior performance with a fraction of the computational footprint.
Succ: Synthesis Success Rate (\%),
Conc: Conciseness Ratio (lower is better),
Acc: Semantic Accuracy (\%),
MCC: Matthews Correlation Coefficient ($\times 100$).
Best results in bold.
}
\label{tab:gpt5_comparison}
\end{table*}

\subsection{Ablation Study}
We analyze the contribution of each component by evaluating five configurations based on \setregex:
(1) \resyn: the full framework with dynamic recursive routing;
(2) w/o \router: a fixed heuristic that rigidly alternates between Segmentation and Partitioning;
(3) \segmenter only: single-level concatenation splitting (equivalent to \splitregex);
(4) \partitioner only: single-level union splitting; and
(5) \setregex: the base model without decomposition.

\begin{table}[t]
    \centering
    \setlength{\tabcolsep}{3.5pt}
    \begin{tabular}{lcccc}
        \toprule
        \multirow{2}{*}{\textbf{Method}} & \multicolumn{2}{c}{\textbf{Configuration}} & \multicolumn{2}{c}{\textbf{RegExLib}} \\
        \cmidrule(lr){2-3} \cmidrule(lr){4-5}
         & \textbf{Strategy} & \textbf{Recursion} & \textbf{Succ} & \textbf{MCC} \\
        \midrule
        \ourrow\textbf{\resyn} & Adaptive & Yes & \underline{68.26} & \textbf{58.97} \\
        w/o \router & Fixed & Yes & \textbf{73.63} & 49.34 \\
        \segmenter only & Seg. Only & No & 65.30 & \underline{58.81} \\
        \partitioner only & Part. Only & No & 40.87 & 37.44 \\
        \setregex & - & No & 38.93 & 36.08 \\
        \bottomrule
    \end{tabular}
    \caption{Ablation study on the \regexlib benchmark quantifying the impact of recursive modules.
    Succ: Synthesis Success Rate (\%), MCC: Matthews Correlation Coefficient ($\times 100$). Best results in bold, second best underlined.}
    \label{tab:ablation}
\end{table}

\paragraph{Impact of Decomposition Strategies.}
Table~\ref{tab:ablation} summarizes the results on the \regexlib benchmark.
Comparing the non-recursive baselines, \segmenter only (65.30\%) significantly outperforms \partitioner only (40.87\%), confirming that Concatenation is the dominant top-level operator in regexes.
However, the full \resyn framework achieves a higher MCC (58.97) than \segmenter only (58.81).
This marginal yet crucial gain (0.16) demonstrates that while segmentation handles most structures, the \partitioner is essential for capturing the semantic branching of Union operators, which a Concatenation-only approach fails to represent.

To understand this marginal 0.16 MCC gain despite \resyn's structural superiority, consider the following target regex:
\begin{itemize}[leftmargin=*, nosep]
    \small
    \item \textbf{Target:} \texttt{\textbackslash d\{1,2\}|\textbackslash d\{1,2\}\textbackslash,\textbackslash d\{1,3\}|\textbackslash x7f}
    \item \textbf{\resyn:} \texttt{[2-9]\textbackslash d?|[1-8]\textbackslash d\textbackslash,\textbackslash d\{1,3\}|\textbackslash x7f}
    \item \textbf{\segmenter only:} \texttt{\textbackslash d?\textbackslash d?[\textbackslash,\textbackslash x7f]?(\textbackslash d\{3\})?}
\end{itemize}
While \resyn correctly deduces the structural branches, suffering a slight MCC penalty (81.65) due to minor digit-range overfitting, \segmenter only collapses the Union, over-generalizing to a malformed regex yet paradoxically achieving a higher MCC (90.45).
This occurs because edit-distance-based negatives lack structural edge cases.
Since mitigating this with computationally prohibitive techniques like Rational Speech Act (RSA) is infeasible for massive datasets, the seemingly small MCC difference actually masks a significant leap in structural correctness.

\paragraph{Efficacy of the Learnable Router.}
A critical insight arises from comparing \resyn with the fixed heuristic w/o \router.
While the fixed strategy achieves the highest synthesis success rate (73.63\%), its MCC drops drastically to 49.34 ($-9.63$ compared to \resyn).
This discrepancy reveals that the fixed heuristic suffers from over-decomposition: it aggressively fragments the problem into trivial sub-problems, satisfying the training examples but destroying the underlying generalization structure.
In contrast, the learned \router in \resyn successfully avoids spurious splits, balancing solvability (Success Rate) with structural integrity (MCC).

\subsection{Case Study and Limitations}
We present a qualitative case study to highlight where \resyn's recursive decomposition excels, followed by a discussion on current limitations in evaluation metrics.

\subsubsection*{Union Structure within Concatenation}
\begin{itemize}[leftmargin=*]
    \small
    \item \textbf{Ground Truth:} \texttt{(I\{2,10\}|V\{2,10\})[A-Z]\{3,4\}}
    \item \textbf{\resyn:} \texttt{(I\{2,10\}|V\{2,10\})[A-Z]\{3,4\}}
    \item \textbf{\texttt{gpt-oss-120b}:} \texttt{(V\{2,\}[A-Z]*V?|I\{2,\}[A-Z]*I?)}
\end{itemize}

\noindent \textbf{Analysis:} This case highlights the efficacy of the \partitioner.
First, the framework decomposes the problem into a prefix and a suffix via Concatenation.
While the suffix is uniform (\texttt{[A-Z]\{3,4\}}), the prefix contains distinct subgroups.
\resyn successfully partitions the prefix into `I-based' and `V-based' clusters, synthesizing the correct Union logic.
In contrast, the baseline fails to disentangle the groups, resulting in an overfitted and convoluted regex.

\subsubsection*{Limitation and Future Work}
Current evaluation metrics like MCC can penalize structurally correct models because negative examples lack structural edge cases, allowing malformed regexes to score highly.
Future work will explore semantic hard-negative mining to better evaluate and unlock \resyn's full potential, moving beyond simple edit-distance perturbations.

\section{Conclusions}
In this work, we bridged the gap between existing neural regex synthesizers and the complexity of real-world regular expressions.
Our statistical analysis highlighted that previous benchmarks fail to reflect the intricate, nested structures found in practical applications.
To overcome the limitations of sequential generation models on such data, we proposed \resyn, a novel framework that combines a robust base synthesizer with a recursive decomposition strategy.

Our contributions are threefold.
First, we introduced a rigorous data processing pipeline, including regex canonicalization and structural de-duplication, ensuring high-quality training and fair evaluation.
Second, we designed \setregex, which utilizes a hierarchical encoder to model the set-based nature of PBE inputs, achieving state-of-the-art efficiency with $30\times$ fewer parameters than baselines.
Third, we demonstrated that our recursive framework, driven by a learnable \router, effectively solves the NP-hard decomposition problem.
Empirically, while non-recursive methods struggle with the deep nesting of real-world regexes, our recursive approach adapts the problem size to the model's capability, yielding significant gains in challenging scenarios.

We believe our divide-and-conquer methodology offers a generalizable direction for program synthesis beyond regular expressions, particularly when target programs exhibit recursive structures.

\section*{Acknowledgements}
Sang-Ki Ko and Seongmin Kim were supported by the National Research Foundation of Korea (NRF) grant funded by the Korean government (MSIT)~(No. RS-2024-00456065).
Yo-Sub Han and Su-Hyeon Kim were supported by the NRF grant~(No. RS-2025-00562134) and the AI Graduate School Program~(No. RS-2020-II201361) funded by the Korean government. 

\bibliographystyle{named} % named.bst
\bibliography{main} % main.bib

\clearpage
\appendix
\renewcommand{\thefigure}{A\arabic{figure}}
\setcounter{figure}{0}
\renewcommand{\thetable}{A\arabic{table}}
\setcounter{table}{0}
\renewcommand{\theequation}{A\arabic{equation}}
\setcounter{equation}{0}
\section{Formal Definition of Regular Expressions}
\label{sec:regex_def}
In this section, we formally define regular expressions over an alphabet $\Sigma$:
\begin{itemize}
    \item \textbf{Empty string:} $\lambda$ is a regex with $L(\lambda) = \{\lambda\}$.
    \item \textbf{Base case:} For $a \in \Sigma$, $a$ is a regex with $L(a) = \{a\}$.
    \item \textbf{Concatenation:} $R_1 R_2$ denotes $\{ xy \mid x \in L(R_1), y \in L(R_2) \}$.
    \item \textbf{Union:} $R_1 | R_2$ denotes $L(R_1) \cup L(R_2)$.
    \item \textbf{Kleene star:} $R^*$ denotes $\bigcup_{k=0}^\infty L(R)^k$.
    \item \textbf{Positive closure:} $R^+$ denotes $L(R^*) \setminus \{\lambda\}$.
    \item \textbf{Option:} $R?$ denotes $\{\lambda\} \cup L(R)$.
    \item \textbf{Interval:} $R\{m,n\}$ denotes $\bigcup_{k=m}^n L(R)^k$.
    \item \textbf{Character classes:} $[C]$ where $C \subseteq \Sigma$ denotes a set of characters such that $L([C]) = C$. This includes ranges $[c_1-c_2]$ representing all characters $c \in \Sigma$ between $c_1$ and $c_2$ in lexicographical order.
    \item \textbf{Named classes:} Shorthand aliases for predefined character sets: $L(\texttt{\textbackslash d})$ is the set of digits $\{0, \dots, 9\}$, $L(\texttt{\textbackslash w})$ is the set of alphanumeric characters (including underscores), and $L(\texttt{.})$ is the set of all symbols in $\Sigma$.
\end{itemize}
We exclude non-regular constructs such as backreferences and look-around assertions.

\section{Regex Canonicalization Pipeline}\label{sec:canonicalization}
\subsection{Filtering Invalid Expressions}
We excluded regular expressions that could not be parsed by the Python \texttt{re} module to ensure compatibility.
Furthermore, we filtered out expressions containing features that require backtracking or exceed the scope of regular languages.
Specifically, expressions with \emph{lookaround assertions} (look-ahead/look-behind) and \emph{backreferences} were removed.
We also discarded expressions containing non-printable ASCII characters.

\subsection{AST Optimization}
To structurally manipulate regular expressions, we defined an Abstract Syntax Tree (AST) composed of six distinct node types, as summarized in Table~\ref{tab:ast-structure}.
This intermediate representation allows for modular optimization.
\begin{table}[ht]
\centering
\begin{tabularx}{\linewidth}{llL}
\toprule
\textbf{Category} & \textbf{Node Type} & \textbf{Description} \\
\midrule
\multirow{3}{*}{Atomic} & \textit{Empty} & The empty string ($\lambda$) \\\cmidrule{2-3}
& \textit{Literal} & A sequence of characters \\\cmidrule{2-3}
& \textit{Character Class} & A set of allowed characters \\
\midrule
\multirow{3}{*}{Operator} & \textit{Concat} & Concatenation of a sequence of child nodes \\\cmidrule{2-3}
& \textit{Union} & Logical alternation of child nodes \\\cmidrule{2-3}
& \textit{Repetition} & Quantification with min/max repetition count \\
\bottomrule
\end{tabularx}
\caption{Abstract Syntax Tree (AST) node types used in regex canonicalization.}
\label{tab:ast-structure}
\end{table}

Following the construction of the initial regex AST, we applied a rule-based \textit{Optimizer}.
The optimizer iteratively applies a set of rewrite rules until the AST reaches a fixed point where no further changes occur.
This ensures the expressions are in their most compact and canonical form.
Finally, expressions that resulted in an empty string after optimization were discarded.
The optimization rules are summarized in Table~\ref{tab:optimization-rules}.
\begin{table}[t]
\centering
\small
\setlength{\tabcolsep}{1pt}
\begin{tabularx}{\columnwidth}{@{}cL@{}}
\toprule
\textbf{Transformation} & \textbf{Description / Rationale} \\
\midrule
\multicolumn{2}{l}{\textit{\textbf{Node Simplification \& Conversion}}} \\
\midrule
$[a] \to a$ & Promote single-character classes to literal nodes. \\
$a \mid b \to [ab]$ & Convert union of single-character literals to a character class. \\
$R\{0,0\} \to \lambda$ & Remove repetitions with zero bounds; simplify to empty string. \\
$R\{1,1\} \to R$ & Remove redundant identity repetitions. \\
$a\{n,n\} \to a \dots a$ & Unroll fixed-length repetitions of literals for structural simplicity. \\
\midrule
\multicolumn{2}{l}{\textit{\textbf{Structural Flattening \& Factoring}}} \\
\midrule
$(R_1 R_2)R_3 \to R_1 R_2 R_3$ & Flatten nested concatenations and unions into a single level. \\
$ab \mid ac \to a(b \mid c)$ & Factor out common prefixes to reduce AST depth and redundancy. \\
$R \lambda \to R$ & Remove identity elements ($\lambda$) from concatenations. \\
$``a" ``b" \to ``ab"$ & Merge consecutive literal nodes into a single string. \\
$R_j \mid R_i \to R_i \mid R_j$ & Sort children of Union lexicographically ($i < j$) for a unique form. \\
\midrule
\multicolumn{2}{l}{\textit{\textbf{Normalization \& Canonicalization}}} \\
\midrule
$R\{m,n\} \to R\{m,\min(n,K)\}$ & Cap upper bounds at threshold $K=10$ to prevent state explosion. \\
$[\neg C] \to [\Sigma \setminus C]$ & Standardize negated classes to an explicit positive representation. \\
$\hat{} R \$ \to R$ & Remove start/end anchors in full-match evaluation settings. \\
$R*?, R*+ \to R*$ & Treat all quantifier variants as standard greedy quantifiers. \\
\bottomrule
\end{tabularx}
\caption{Detailed AST optimization rules for regex canonicalization.}
\label{tab:optimization-rules}
\end{table}

\subsection{Literal Anonymization}
Our proposed framework relies on recursive splitting, where the decision to split depends on the boundaries of sub-expressions rather than their specific textual content.
To force the model to focus on these structural boundaries and to reduce the vocabulary size, we anonymized literals with a length of 2 or greater.
These literals were replaced with special \textit{Anonymization Tokens}.
This abstraction ensures that the model learns to treat long literals as atomic blocks, preventing it from overfitting to specific keywords or character sequences found in the training data.
The anonymization tokens were selected from non-printable ASCII ranges to avoid conflict with standard regex operators and printable characters.
The available ranges are $[3, 8]$, $[14, 31]$, and $\{127\}$. We reserved specific indices ($0, 1, 2$) for special system tokens: \texttt{PAD}, \texttt{Empty String} ($\lambda$), and \texttt{Separator/Empty Set} ($\emptyset$), respectively.
To ensure that the neural network learns the \textit{role} of an anonymized literal rather than associating a specific token ID with a specific pattern, we randomize the mapping between literals and tokens for each training instance.
This prevents the model from overfitting to specific token embeddings and encourages learning structural relationships instead.

\subsection{Serialization}
Converting the AST back into a regular expression string requires careful handling of operator precedence and canonicalization. We implemented a serialization method that:
(\textit{i}) adds parentheses strictly when necessary based on operator precedence (e.g., \texttt{(a|b)*} when \textit{Union} is a child of \textit{Repetition} or \textit{Concat}), omitting redundant ones to maintain brevity;
(\textit{ii}) serializes \textit{Character Class} nodes as concisely as possible through a two-step process: first, the character set is matched against predefined named character classes (e.g., \texttt{\textbackslash d} for digits, \texttt{\textbackslash w} for word characters, \texttt{\textbackslash s} for whitespace); if no exact match is found, the algorithm identifies consecutive ASCII characters and groups them into range notation (e.g., \texttt{a-z}), while characters that do not belong to any consecutive sequence are retained as individual literals, resulting in compact representations such as \texttt{[0-9\_a-z]} or \texttt{[0-2a-c]};
(\textit{iii}) simplifies \textit{Repetition} bounds into standard quantifiers when possible: $\{0, \infty\} \to \texttt{*}$, $\{1, \infty\} \to \texttt{+}$, $\{0, 1\} \to \texttt{?}$, and $\{n, n\} \to \{n\}$, while general cases that do not match these patterns remain in the explicit bound notation $\{m, n\}$.
\section{Detailed Algorithm and Running Example of \resyn}
\label{app:algorithm}
The inference process of \resyn, formally detailed in Algorithm~\ref{alg:resyn} and illustrated in Figure~\ref{fig:flowchart}, proceeds recursively starting from the initial set $P$ and $N$.
First, for the trivial case where the set contains only a single positive example ($|P|=1$), we bypass the neural components and directly return the escaped literal string.
This is because inductive generalization requires at least two examples to identify shared patterns; without them, the literal string is the only unambiguous and correct solution.

For non-singleton sets, the \router evaluates the structural complexity.
If the set is simple enough, an off-the-shelf regex synthesizer generates the solution.
To enhance robustness, we incorporate a heuristic fallback mechanism: if the neural synthesizer fails to produce a valid regex consistent with the examples, we iteratively test a predefined list of common patterns (refer to Appendix~\ref{app:fallback_patterns}) and return the valid one describing the smallest language.
Crucially, the viability of this heuristic is intrinsically enabled by our recursive decomposition strategy.
While such simple patterns are insufficient to describe the complex global structure of real-world regexes, they are highly effective at characterizing the sub-components isolated at the leaf nodes of our derivation tree.
This prevents the entire derivation tree from being invalidated due to a failure in a single leaf node.

Otherwise, the set is decomposed by either the \partitioner or \segmenter, and the process repeats for each resulting sub-problem.
To prevent spurious decomposition, we forbid applying the same decomposition strategy consecutively.
Furthermore, to avoid trivial solutions that merely enumerate examples without inductive generalization, we restrict the partitioning mechanism.
Specifically, if a partitioning operation yields only singleton subsets, we reject this decomposition and directly invoke the base synthesizer on the current set.
This mechanism prevents the framework from overfitting to individual strings and forces the model to discover generalized patterns over the collective examples.
This results in a derivation tree where internal nodes represent decomposition operations (Union or Concatenation) and leaf nodes represent sub-regex synthesis.
Finally, the partial regexes generated at the leaves are composed bottom-up according to the operator types of the internal nodes to produce the final complete regular expression (see Figure~\ref{fig:running_example} for a running example).

\begin{figure*}[ht!]
    \centering
    \includegraphics[width=\linewidth]{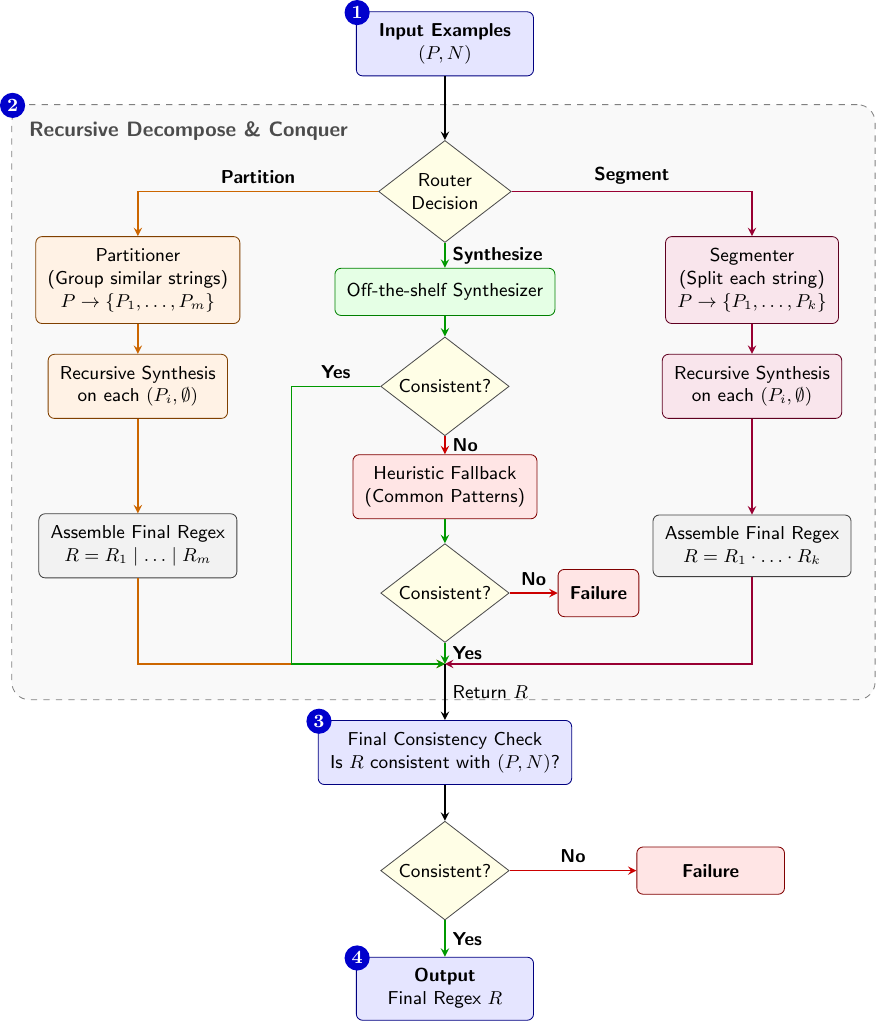}
    \caption{The overall workflow of the \resyn framework. It illustrates the recursive decompose-and-conquer process, adaptively routing the input examples to the Partitioner, Segmenter, or Synthesizer to assemble the final regular expression.}
    \label{fig:flowchart}
\end{figure*}

\begin{figure*}[t]
  \centering
  \includegraphics[width=\linewidth]{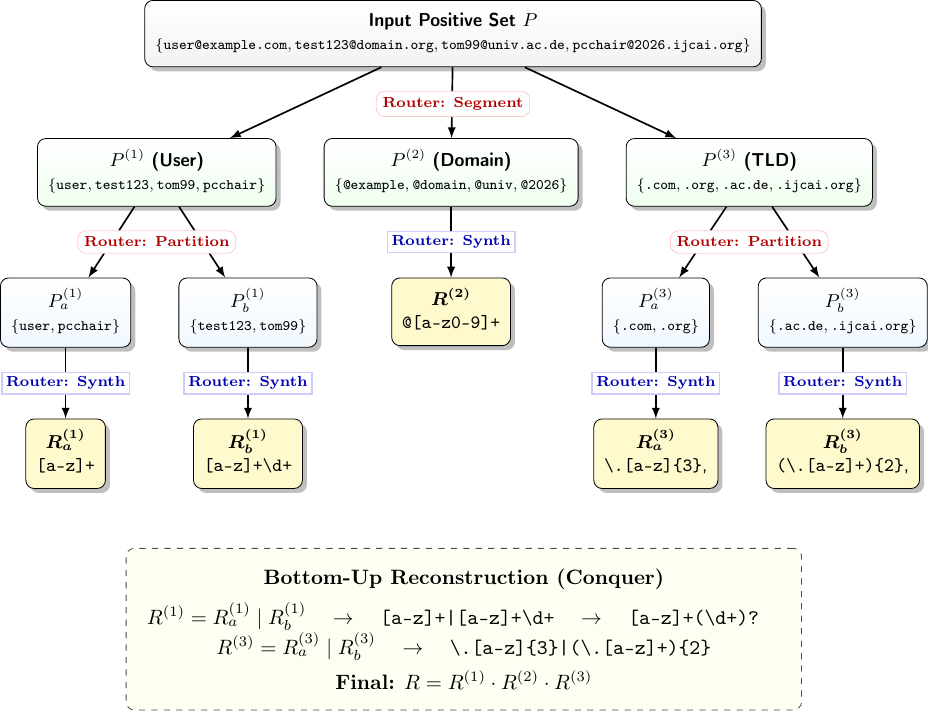}
  \caption{A running example of the \resyn framework process.
  The input positive set $P$ is first decomposed via Segmentation into three logical components ($P^{(1)}$: User, $P^{(2)}$: Domain, $P^{(3)}$: TLD), mirroring a Concatenation structure.
  Subsequently, the \router recursively determines the strategy for each subset.
  For instance, $P^{(1)}$ and $P^{(3)}$ are further decomposed via Partitioning (Union), while non-decomposable components are synthesized directly.
  Finally, the partial regexes are reconstructed bottom-up to form the complete regex.
  Note that this example is simplified for illustrative purposes.}
  \label{fig:running_example}
\end{figure*}

\begin{algorithm}[ht]
\caption{\resyn Recursive Synthesis}
\label{alg:resyn}
\small
\begin{algorithmic}[1]
\Procedure{Synthesize}{$P, N$}
    \State $R \gets$ \Call{RecursiveSynthesize}{$P, N, \textsc{null}$}
    \If{$R = \textsc{Failure}$ \textbf{or} \textbf{not} \Call{Consistent}{$R, P, N$}}
        \State \Return \textsc{Failure}
    \EndIf
    \State \Return $R$
\EndProcedure
\State
\Procedure{SynthesizeFromSingleton}{$P$}
    \State $w \gets$ the single string in $P$
    \If{$w = \lambda$}
        \State \Return a regex matching only the empty string
    \Else
        \State \Return \Call{Escape}{$w$} \Comment{Literal with special chars escaped}
    \EndIf
\EndProcedure
\State
\Procedure{SynthesizeWithFallback}{$P, N$}
    \State $R \gets$ \Call{BaseSynthesizer}{$P, N$}
    \If{$R \neq \textsc{Failure}$ \textbf{and} \Call{Consistent}{$R, P, N$}}
        \State \Return $R$
    \EndIf
    \State $\mathcal{H} \gets \{ \texttt{\textbackslash d+}, \texttt{\textbackslash d*}, \texttt{[a-z]+}, \dots \}$ \Comment{Common patterns}
    \State $R_{best} \gets \textsc{Failure}$
    \For{$h \in \mathcal{H}$}
        \If{\Call{Consistent}{$h, P, N$} \textbf{and} \Call{IsSmaller}{$h, R_{best}$}}
            \State $R_{best} \gets h$
        \EndIf
    \EndFor
    \State \Return $R_{best}$
\EndProcedure
\State
\Procedure{RecursiveSynthesize}{$P, N, \textit{prevType}$}
    \If{$|P| = 1$}
        \State \Return \Call{SynthesizeFromSingleton}{$P$}
    \EndIf
    \State $a \gets$ \Call{Router}{$P$} \Comment{Predict decomposition strategy}
    \If{$a = \textsc{Seg}$ \textbf{and} $\textit{prevType} \neq \textsc{Seg}$}
        \State $(P_1, \ldots, P_k) \gets$ \Call{Segmenter}{$P$}
        \If{$k > 1$}
            \For{$i = 1$ \textbf{to} $k$}
                \State $R_i \gets$ \Call{RecursiveSynthesize}{$P_i, \emptyset, \textsc{Seg}$}
            \EndFor
            \State \Return $R_1 \cdot R_2 \cdots R_k$
        \EndIf
    \EndIf
    \If{$a = \textsc{Part}$ \textbf{and} $\textit{prevType} \neq \textsc{Part}$}
        \State $\{P_1, \ldots, P_m\} \gets$ \Call{Partitioner}{$P$}
        \If{$m > 1$}
            \If{\textbf{all} $|P_i| = 1$ \textbf{for} $i \in \{1, \dots, m\}$} 
                \State \textbf{return} \Call{SynthesizeWithFallback}{P, N}
            \EndIf
            \For{$i = 1$ \textbf{to} $m$}
                \State $R_i \gets$ \Call{RecursiveSynthesize}{$P_i, \emptyset, \textsc{Part}$}
            \EndFor
            \State \Return $R_1 | R_2 | \cdots | R_m$
        \EndIf
    \EndIf
    \State \Return \Call{SynthesizeWithFallback}{$P, N$}
\EndProcedure
\end{algorithmic}
\end{algorithm}

\section{List of Common Patterns for Fallback Strategy}
\label{app:fallback_patterns}
To synthesize atomic components that the neural model fails to resolve, we employ a deterministic fallback mechanism.
This mechanism iterates through a predefined list of base character classes $\mathcal{B}$, combining each with quantifiers $Q \in \{+, *\}$.
The regex $R = b \cdot q$ (where $b \in \mathcal{B}, q \in Q$) is validated against the example sets, and the first consistent pattern is returned.

Crucially, the search order is designed to favor specificity over generality. As shown in Table~\ref{tab:fallback_patterns}, the algorithm prioritizes narrow character sets (e.g., digits, specific alphabets) before checking broader classes (e.g., word characters) or the wildcard (dot). This ensures that the synthesized regex captures the smallest language necessary to describe the examples, avoiding over-generalization.

\begin{table}[!ht]
    \centering
    \begin{tabular}{cll}
        \toprule
        \textbf{Priority} & \textbf{Base Pattern} & \textbf{Description} \\
        \midrule
        1 & \texttt{\textbackslash d} & Digits (\texttt{[0-9]}) \\
        2 & \texttt{[a-z]} & Lowercase alphabets \\
        3 & \texttt{[A-Z]} & Uppercase alphabets \\
        4 & \texttt{[a-zA-Z]} & Mixed case alphabets \\
        5 & \texttt{[0-9a-fA-F]} & Hexadecimal characters \\
        6 & \texttt{\textbackslash w} & Word characters \\
        7 & \texttt{\textbackslash s} & Whitespace characters \\
        \midrule
        8 & \texttt{\textbackslash D} & Non-digits \\
        9 & \texttt{\textbackslash W} & Non-word characters \\
        10 & \texttt{\textbackslash S} & Non-whitespace characters \\
        11 & \texttt{.} & Any character (Wildcard) \\
        \bottomrule
    \end{tabular}
    \caption{Prioritized list of base patterns used in the fallback mechanism. The algorithm searches strictly in this order, combining each pattern with quantifiers (`+' then `*').}
    \label{tab:fallback_patterns}
\end{table}
\section{Details of Neural Decomposition Modules}
\label{app:modules}

In this section, we provide the detailed architectures and training mechanisms for the three learnable modules constituting the \resyn framework.

\subsection{\router}
The \router acts as a central policy network that determines the optimal decomposition strategy for a given set of positive examples $P$.
While it shares the Hierarchical Encoder architecture with \setregex to capture structural features, it incorporates two key modifications tailored for its control logic.
First, the \router takes only the positive examples $P$ as input, as the decomposition strategy relies solely on the structural patterns of the strings to be matched, thereby eliminating the need for type embeddings to distinguish negative examples.
Second, it utilizes only the final set-level context vector $c$ to compute the policy.

Formally, the \router maps the context vector $c$ to a probability distribution over three discrete actions: \texttt{Synthesize}, \texttt{Partition}, and \texttt{Segment}.
\begin{equation}
    \pi(a | P) = \text{Softmax}(\text{MLP}(c))
\end{equation}
where $a \in \{\texttt{Synthesize}, \texttt{Partition}, \texttt{Segment}\}$.
During inference, if \texttt{Synthesize} is selected, an off-the-shelf regex synthesizer is invoked to generate the base case regex.
Otherwise, the corresponding decomposition module (\partitioner or \segmenter) is triggered to solve the problem recursively.

To train the \router, we derive ground-truth labels directly from the target regex's AST.
We inspect the top-level operator of the target regex and assign the labels as follows:
\texttt{Partition} for a Union operator, \texttt{Segment} for a Concatenation operator, and \texttt{Synthesize} for any other operators (non-Union/Concat).

\subsection{\partitioner}
The \partitioner is responsible for dividing the example set $P$ into $m$ disjoint subsets $\{P_1, \ldots, P_m\}$.
It aims to cluster strings that share similar structural patterns identifiable by a common sub-regex.
Like the \router, the \partitioner takes only positive examples as input, focusing solely on the structural similarities among the strings.
The architecture adapts the Hierarchical Encoder of \setregex with two specific modifications tailored for the clustering task:
\begin{itemize}
    \item \textbf{Ordered String Encoding:} Unlike the permutation-invariant set encoder, the \partitioner respects the input order. We add positional encodings to the string-level embeddings $\mathbf{h}_i$ to differentiate distinct examples. No set-level context vector is generated, as the focus is on individual string assignment.
    \item \textbf{Pointer-based Decoding:} The decoder employs a mechanism inspired by Pointer Networks~\cite{vinyals2015pointer} to predict a sequence of cluster assignments corresponding to the input strings.
\end{itemize}

The decoder operates auto-regressively to determine the cluster indices.
It begins with a fixed start token (representing cluster 0) and, at each step, uses the predicted cluster label from the previous step as input to generate the assignment for the current string.
To supervise this process, we derive labels from regexes where the top-level operator is a Union ($R = R_1 | \dots | R_k$).
For each string $w \in P$, we identify the first sub-regex $R_i$ that matches $w$ and group the strings accordingly.
Crucially, to align with the decoding mechanism, these absolute group mappings are converted into relative labels:
the first string is always assigned to cluster 0, and each subsequent string is assigned either an existing cluster index (if it matches the same sub-regex as a previous string) or a new index to initiate a new cluster.
This approach allows the \partitioner to dynamically determine the optimal number of subsets $m$ based on data complexity, rather than being restricted to a fixed number of splits.

\subsection{\segmenter}
The \segmenter determines the optimal split points within the strings to decompose a concatenation problem.
Unlike the previous modules that utilize hierarchical encoding, the \segmenter adopts a standard Transformer encoder-decoder architecture~\cite{vaswani2017attention}.
This design choice is motivated by the need for fine-grained, character-level alignment across all examples simultaneously.
Critically, prior approaches like \splitregex~\cite{JALC-2025-157} process strings independently, which often leads to structural inconsistencies; our model addresses this by processing all examples jointly.

The input is formed by concatenating all characters from the strings in $P$, separated by distinct delimiter tokens.
The model operates as an auto-regressive sequence labeler on this flattened sequence, predicting a segment index for each character.
To provide supervision for this task, we utilize regexes where the top-level operator is a Concatenation ($R = R_1 \cdot R_2 \dots R_k$).
For every string $w \in P$, we decompose it into substrings $w = w_1 \dots w_k$ (where each $w_i \in L(R_i)$) and assign each character the index $i$ of its corresponding segment.
By training on these labels, the decoder learns to attend to both the global context and previously predicted indices, ensuring that the segmentation decisions remain consistent across all examples.
Finally, the strings are split according to these predicted indices, forming a sequence of new sub-problems $(P_1, P_2, \ldots, P_k)$ to be solved recursively.
\section{NP-hardness of Optimal Example Decomposition and Concise Regex Synthesis Problem}
\label{appendix:full_complexity}
We now establish that the problem of optimally decomposing a set of examples using the concatenation or union operations is \textit{NP-hard},
and thus this problem is intractable with deterministic algorithms, unless $P = NP$.

First, we need a quantitative criterion for the conciseness of regexes among several candidates
to identify an optimal one.
Here we define a cost for regexes by the number of symbols in the expression,
and a concise regex will have the lowest such cost.
\begin{definition}[Expression cost]
Let $R$ be a regex over $\Sigma$.
The expression cost~$c_{E}(R)$ of $R$
is the minimum number of symbols, except operators, in $R$.
\end{definition}
Although there are various complexity criteria for regexes and/or regular languages,
such as the star height~\cite{Eggan1963} of the regexes or the minimum number of DFA states for the languages,
we use the number of symbols for the cost
since the length of a regex
is bounded above by a constant factor (not exceeding 4)
to the number of symbols
presented in the expression.
This is easily proved
using the fact that the syntax tree of regexes
is a binary tree.

\begin{example}
Here are some regexes and their expression costs.
\begin{itemize}
\item $c_{E}(\texttt{a}) = 1$
\item $c_{E}(\texttt{a|ab|ac}) = 5$
\item $c_{E}(\texttt{a(b|c)?}) = 3$
\end{itemize}
%The third language is often described in $a(b|c)?$ with practical regex notation,
%justifying the choice of the expression cost.
Note that the second and the third expressions describe the same language but have different cost values.
\end{example}

We now define a similar cost for languages.
Here we only consider finite languages
since
(1) our goal is to infer a regex from finite examples, and
(2) not every languages are regular; but finite ones are.

\begin{definition}[Language expression cost]
Let $L$ be a finite language over $\Sigma$.
The language expression cost~$c_{E}(L)$ of $L$
is the minimum expression cost~$c_{E}(R)$,
where the regex~$R$ is for some finite language~$M=L(R)$
that contains $L$.
%and does not contain strings outside $L$ whose supersequence is a proper substring of strings in $L$.
\end{definition}
% We use the term `expression cost'
% to refer to either the expression cost or the language expression cost
% if it is clear in its context.
With these cost functions, we define the concise regex problem
that computes the lowest language expression cost for a given sample set.

For instance, consider a language~$L_1=\{\texttt{apple}, \texttt{apply}\}$.
there are two regex with minimum cost: \texttt{appl(e|y)} and \texttt{apple?y?}.
Note that the second expression
describes a superset of $L_1$ and
still fits in the condition.

\begin{problem}[Concise regex problem]
Let~$S$ be a finite set of strings and $r$ a nonnegative integer.
Decide whether $c_{E}(S) \le r$ or not.
\end{problem}

In addition to that, we need a cost for such decomposition of multiple strings
to identify an optimal split,
\begin{definition}[String decomposition cost] % This is the block subdivision
Let $n$ strings~$w_1, w_2, \ldots, w_n$ be given.
we can represent these strings with
$m$ strings~$x_1, x_2, \ldots, x_m$ and 
$n$ non-decreasing functions~$p_i: [1, l_i] \to [1, m]$ over integers for for some $l_i$ and $1 \le i \le n$.
These must satisfy $w_i = x_{p_i(1)} x_{p_i(2)} \cdots x_{p_i(l_i)}$ for all $1 \le i \le n$.

Then, the string decomposition cost~$c_d(w_1, w_2, \ldots, w_n)$
is the minimum of total length of the strings~$x_1, x_2, \ldots, x_m$.
\end{definition}
% Note that, if a string~$w_i$ has repeated substring,
% $\{x_k\}_k$ could contain duplicated strings
% since $p_i$'s are monotonically increasing
This cost definition describes
the equivalent procedure
to splitting the positive examples
into substrings and identifying
the matching groups.

\begin{example}
For three strings~($w_i$'s) \texttt{bat}, \texttt{cat}, and \texttt{dog},
the string decomposition cost is 7 since
$\{x_k\}_k = \{\texttt{b}, \texttt{c}, \texttt{at}, \texttt{dog}\}$ and
all $w_i$'s have a mapping~$p_i$ to rebuild themselves.
Note that the language expression cost for the strings are also 7 (The regex is \texttt{(b|c)at+dog}).
\end{example}

%\begin{definition}
%\begin{gather*}
%c_d(w_1, w_2, \ldots, w_n) = \\
%\begin{cases}
%  0 & n = 0 \\
%  0 & \{ w_1, \ldots, w_n\} = \{\lambda\} \\
%  |x| & \{ w_1, \ldots, w_n\} \setminus \{\lambda\} = \{x\} \\
%  \min_{w_i = u_iv_i}{c(u_j\ldots) + c(v_j\ldots)} & \text{otherwise}
%\end{cases}.
%\end{gather*}
%\end{definition}

The optimal alignment in this paper
is a symbol-wise decomposition of strings with the minimum alignment cost.
\begin{definition}[Alignment cost] % This is the width-1-restricted block subdivision
An alignment tuple is $(u_1, u_2, \ldots, u_n) \in \Omega = (\Sigma \cup \{\lambda\})^n \setminus \{\lambda\}^n$.
An alignment of $n$ strings~$w_1, w_2, \ldots, w_n$ is
a length-$m$ sequence~$t_1 t_2 \cdots t_m$  of alignment tuples
satisfying
$\{t_{j1}, t_{j2}, \ldots, t_{jn}\} \setminus \{\lambda\} = \{\sigma\}$ where $t_{ij}$ is the $j$th string of the tuple~$t_i$ and $\sigma \in \Sigma$,
and catenation of all $t_{ji}$ for $1 \le j \le m$ yields $w_i$.
The cost of an alignment is its length, $m$.

For a finite language~$S$,
$c(S)$ denotes the minimum alignment cost of strings in $S$.
%\begin{gather*}
%c(w_1, w_2, \ldots, w_n) = \\
%\begin{cases}
%  0 & n = 0 \\   % auto by def
%  0 & \{ w_1, \ldots, w_n\} = \{\lambda\} \\   % impossible by def
%  1 & \{ w_1, \ldots, w_n\} \setminus \{\lambda\} = \{\sigma\}, \sigma \in \Sigma \\  % defined
%  \min_{}{c(u_j\ldots) + c(v_j\ldots)} & \text{otherwise}  % defined
%\end{cases}.
%\end{gather*}
\end{definition}

A notable point is that
the optimal string decomposition cost and 
the optimal alignment cost are always equal
for any finite language.
\begin{lemma}\label{lem:cost-align-decompose}
For any finite language~$S$, $c_d(S) = c(S)$.
\end{lemma}
\begin{proof}
$c_d(S)  \le c(S)$:
Suppose $c(S) = C$.
Since every alignment tuples in the optimal alignment has unique symbol,
we can use those symbols for the decomposed substrings.
Then, it is obvious that $S$ has string decomposition cost of at most $C$.

$c(S) \le c_d(S)$:
Suppose $c_d(S) = C$ with $m$ decomposed substrings~$x_1, \ldots, x_m$.
By definition, $\sum_{k}{|x_k|} = C$.
Then, for each $1 \le k \le m$,
we can create $|x_k|$ alignment tuples of $(u_1, u_2, \ldots, u_n)$
such that, for $j$th tuple, $u_i = x_{kj}$ for $1 \le i \le n$ if the range of $p_i$ contains $k$,
and $\lambda$ otherwise.
With these alignment tuples concatenated in order, we have an alignment for $S$
such that the aligned symbols have a common matching substrings in string decomposition,
%Then, the sequence of those alignment tuples
%is an alignment of the strings in $S$,
showing that $S$ has an alignment of cost~$C$.
\end{proof}
Thus, instead of showing the hardness of
the optimal string decomposition,
we show that the optimal alignment problem is hard.

\begin{problem}[Optimal alignment problem]
Let~$S$ be a finite set of strings and $r$ a nonnegative integer.
Decide whether $c(S) \le r$ or not.
\end{problem}

\begin{definition}[Supersequence]
For two strings~$x$ and $y$ over $\Sigma$,
$y$ is a supersequence of $x$
if $x$ matches against
the regex~$\Sigma^* u_1 \Sigma^* u_2 \Sigma^* \cdots \Sigma^* u_k \Sigma^*$
where $y = u_1 u_2 \cdots u_k$.
$x \preceq y$ denotes that $y$ is a supersequence of $x$.
\end{definition}

\begin{lemma}\label{lem:align-nphard-app}
The optimal alignment problem is {\sf NP}-complete.
\end{lemma}
\begin{proof}
It is obvious that the problem is in {\sf NP}.

For {\sf NP}-hardness,
we reduce the shortest common supersequence problem~\cite{RaihaU1981},
which decides whether the shortest common supersequence for given set of strings is at most $m$ or not.
Suppose we have $n$ strings~$w_1, \ldots, w_n$.
We denote $w_i = w_{i1} w_{i2} \cdots w_{il_i}$, where $w_{ij} \in \Sigma$ for all $1 \le j \le l_i$.

We will first show that, if the shortest common supersequence of the strings has length $m$,
the optimal alignment cost for $n$ strings~$w_1, \ldots, w_n$ is exactly $m$.
Let $s=s_1 s_2 \cdots s_m$ be the supersequence of $n$ strings.
Then, for each string~$w_i$, there is a mapping~$f_i: [1, l_i] \to [1, m]$  % note that $l_i \le m$.
such that $w_{ij} = s_{f_i(j)}$ for all $1 \le i \le n+1$.
Then, there is an inverse mapping~$f_i^{-1}: [1, m] \to [1, l_i] \cup \{\bot\}$ of $f_i$,
where $f_i^{-1}(k) = \bot$ if there is no $j$ such that $f_i(j) = k$.

For example, Consider an example in Figure~\ref{fig:align-nphard-example} showing an alignment of two strings $w_1=\texttt{http}$ and $w_2=\texttt{ftps}$ against their common supersequence~$s=\texttt{hfttps}$.
It is obvious to see that the values of $l_1$ and $l_2$ are 4, since both strings are of length 4.
The alignment matches the four symbols of $w_1$ against the 1st, 3rd, 4th, 5th symbols of the common supersequence~$s$,
thus $f_1(1) = 1$, $f_1(2) = 3$, and so on.
The inverse mapping is obvious, with the missing values $f_1^{-1}(2) = f_1^{-1}(6) = \bot$.
$f_2$, the alignment mapping for $w_2$, is also similarly defined.

Then, it is obvious that the corresponding alignment tuple has symbols of $\{ w_{i, f_i^{-1}(k)} \mid i \in [1, n+1] \} = \{\sigma, \lambda \}$,
where $w_{i\bot} = \lambda$ to denote that the alignment has no matching symbol in $w_i$.
Note that the cost of such alignment tuple is 1.
By summing up all these costs, we get exactly $m$ as an upper bound of their alignment cost.

\begin{figure}
\centering
\begin{tabular}{ccccccccl}
\multicolumn{2}{l}{Index}& 1 & 2 & 3 & 4 & 5 & 6 \\\hline
$s$   & $=$ & \texttt{h} & \texttt{f} & \texttt{t} & \texttt{t} & \texttt{p} & \texttt{s} \\\hline
$w_1$ & $=$ & \texttt{h} &            & \texttt{t} & \texttt{t} & \texttt{p} &            \\
$w_2$ & $=$ &            & \texttt{f} & \texttt{t} &            & \texttt{p} & \texttt{s} \\\hline
\end{tabular}
\caption{Example alignment of two strings `\texttt{http}' and `\texttt{ftps}'
\label{fig:align-nphard-example}}
% note that the cost-efficient regex for this alignment is "(h|f)tt?ps?", having several additional words including https, htps, ftp, fttp, and so on.
% the exact protocol string would be (ht|f)tps?; having the same symbol cost of 6.
\end{figure}

Thus we can solve the shortest common supersequence problem using the optimal alignment problem.
Since the former is {\sf NP}-complete, the optimal alignment problem is {\sf NP}-hard.
\end{proof}

\begin{corollary}
Let~$S$ be a finite set of strings and $r$ a nonnegative integer.
It is {\sf NP}-complete
to decide whether~$c_d(S) \le r$ or not.
\end{corollary}

Moreover, it is possible to show that the
language expression cost for a finite language~$S$
is exactly the string decomposition cost of the same language.
\begin{lemma}\label{lem:cost-re-decompose}
For a finite set~$S$,
$c_{E}(S) = c_d(S)$.
\end{lemma}
\begin{proof}
Let $S = \{w_1, w_2, \ldots, w_n\}$ be given.
It is easy to show that,
if $S$ is either of $\{w\}$ or  $\{\lambda, w\}$,
$c_{E}(S) = c_d(S) = |w|$.

First, suppose the mappings~$p_i$ for $S$'s string decomposition
can be partitioned into $q \ge 2$ collections~$Z_1, \ldots, Z_q$
such that any two mappings in different collection have disjoint ranges.
In other words, each collection~$Z_j$ spells a subset of $S$, exactly $\{ w_i \mid p_i \in Z_j \}$.
Then, there is a corresponding partition of $S$ composed with $q$ subsets~$S_1, \ldots, S_q$,
and $c_{E}(S) = \sum_{j}{c_{E}(S_j)} = \sum_j{c_d(S_j)} = c_d(S)$
if we assume that the claim holds for smaller sets.
Note that the first equation comes from the fact that $L(R_1) \cup L(R_2) = L(R_1 | R_2)$
and all $S_j$ are disjoint,
and the third is due to the disjointness of union of function ranges in each $Z_j$.

Thus, it is enough to consider the case where we cannot partition $S$ as above.
In this case,
we can build a regex~$R$ so that the decomposed substrings~$x_k$~($1 \le k \le m$) appear at most once in $R$:
$R = e_1 e_2 \cdots e_m$,
where $e_j = x_j$ if $j$ is in the range of every $p_i$ ($1 \le i \le n$)
and $e_j = x_j\texttt{?}$ otherwise.
% if e_i = e_i+1, e_i = e_i*, e_j = lambda 
Then, the language expression cost is bounded above by the sum of lengths of decomposed substrings: $c_{E}(S) \le c_{E}(R) = \sum_k\{|x_k|\}$.
Moreover, it is also true that the substrings~$x_j$ must appear at least once.
Thus, the language expression cost~$c_{E}(S)$ and
the string decomposition cost~$c_d(S)$ are equal.
\end{proof}

From the above lemmas,
we can finally show that the concise regex problem is intractable in practice.
\begin{theorem}
The problem of inferring a concise regex by decomposition is {\sf NP}-complete.
\end{theorem}
\begin{proof}
The optimal alignment problem is {\sf NP}-complete (Lemma~\ref{lem:align-nphard-app}) and $c_{E}(S) = c_d(S) = c(S)$ (Lemmas~\ref{lem:cost-align-decompose} and \ref{lem:cost-re-decompose}).
%computing $c_{E}(S)$ is {\sf NP}-complete.
\end{proof}
Also note that this is a special case of the problem tackled with \RecuRegex{}
with the empty negative set.
\section{Detailed Dataset Construction and Example Generation}
\label{app:data_details}

\paragraph{Canonicalization and Structural Filtering.}
All collected regexes were first standardized using our Regex Canonicalization pipeline (see Section~\ref{sec:canonicalization}). To prevent data leakage, we enforced a strict separation between training and evaluation sets based on structural equivalence (AST signatures), ignoring specific literals. Any regex in the training set sharing the same structure as those in the test benchmarks was removed.

\paragraph{Sub-Regex Extraction.}
To augment the training data, we recursively extracted all valid sub-regexes from complex expressions. For example, from \texttt{($R_1$|$R_2$)*}, we included the original and all constituent sub-expressions \{\texttt{($R_1$|$R_2$)}, \texttt{$R_1$}, \texttt{$R_2$}, \dots\}, increasing dataset size by about 28\%.

\paragraph{Filtering Criteria.}
\begin{itemize}
    \item \textbf{Length:} Regexes exceeding 110 characters were excluded.
    \item \textbf{Complexity:} Regexes with a top-level Union operator having more than 10 branches were discarded.
\end{itemize}

\paragraph{Train/Validation Split.}
After filtering, the dataset was split into training and validation sets (9:1 ratio), maintaining AST signature disjointness.

\subsection{Example Generation}

\paragraph{Positive Examples.}
For each regex, up to 10 positive examples were generated using the \texttt{IntXeger}\footnote{\url{https://github.com/k15z/IntXeger}} library, with a 60-second timeout per regex. At least 2 positive examples were required. For Union operators, at least one string per branch was sampled; for unbounded repetitions, the maximum count was limited to 20.

\paragraph{Negative Examples.}
Negative examples were generated using a mutation-based strategy: random edit operations (insertion, deletion, substitution) were applied to positive examples.
Some regexes (e.g., \texttt{.*}) may have no negative examples.

\paragraph{Substring Expansion.}
Positive example strings were recursively split according to the regex structure to generate substring sets, used for training decomposition modules. This expansion was applied only when the top-level operator was Concat or Union.

\paragraph{Hold-out Set and Redundancy Removal.}
For evaluation, an additional disjoint set of examples (Hold-out set) was generated. Internal redundancy in the test set was eliminated by filtering overlapping AST signatures, ensuring each test instance represents a unique structural pattern.

\paragraph{Final Dataset Statistics.}
Table~\ref{tab:comprehensive_stats} summarizes the comprehensive statistics of the final datasets constructed through these processes, detailing the distribution of top-level AST node types, substring expansion ratios, and the properties of the generated positive and negative strings.

\begin{table*}[ht]
\centering
\begin{tabular}{@{}llccccc@{}}
\toprule
\multirow{2.5}{*}{\textbf{Category}} & \multirow{2.5}{*}{\textbf{Metric / Feature}} & \multirow{2.5}{*}{\textbf{Train}} & \multirow{2.5}{*}{\textbf{Valid}} & \multicolumn{3}{c}{\textbf{Test Benchmarks}} \\ \cmidrule(l){5-7} 
 & & & & \textbf{\sregex} & \textbf{\snort} & \textbf{\regexlib} \\ \midrule
\multicolumn{2}{@{}l}{Instance Count (Unique)} & 1.4M & 32K & 334 & 352 & 1,752 \\ \midrule
\multicolumn{2}{@{}l}{Substring Ratio} & 63.71\% & 1.35\% & 0.00\% & 0.00\% & 0.00\% \\ \midrule
\multirow{4}{*}{Top-level Op.}
& Character Class & 12.32\% & 6.74\% & 0.00\% & 0.00\% & 1.77\% \\ 
& Concat & 30.90\% & 87.69\% & 99.10\% & 98.01\% & 74.83\% \\
& Repetition & 53.38\% & 2.86\% & 0.60\% & 1.99\% & 6.45\% \\ 
& Union & 3.40\% & 2.71\% & 0.30\% & 0.00\% & 16.95\% \\ 
\midrule
\multirow{3}{*}{AST Depth} & Mean & 2.47 & 3.28 & 3.38 & 3.16 & 3.94 \\ 
& Median & 2.0 & 3.0 & 3.0 & 3.0 & 4.0 \\
& Max Depth & 25 & 19 & 5 & 7 & 10 \\
\midrule
\multirow{2}{*}{Positive Str} & Count (Mean) & 8.16 & 9.63 & 9.91 & 9.70 & 9.88 \\
& Length (Mean) & 14.35 & 23.97 & 15.61 & 27.36 & 24.40 \\
\midrule
\multirow{2}{*}{Negative Str} & Count (Mean) & 3.62 & 9.83 & 10.00 & 10.00 & 9.98 \\
& Length (Mean) & 22.31 & 23.99 & 16.44 & 27.19 & 25.26 \\
\bottomrule
\end{tabular}
\caption{Comprehensive Dataset Comparison: Train, Validation, and Test Benchmarks breakdown. The structural gap between synthetic (\sregex) and real-world (\regexlib) datasets is highlighted by the Top-level Operator distribution and AST Depth.}
\label{tab:comprehensive_stats}
\end{table*}

\section{Prompt Engineering Details}
\label{app:prompt_details}
We utilized a unified zero-shot prompting strategy for both the \texttt{gpt-oss-120b} and \texttt{GPT-5} baselines, as shown in Figure~\ref{fig:prompt_template}.
The template includes strict technical constraints to ensure compatibility with the \texttt{google-re2} engine.

\begin{figure*}[t]
    \centering
    \begin{tcolorbox}[colback=gray!5!white, colframe=black, boxsep=1mm, left=1mm, right=1mm, top=1mm, bottom=1mm, title=\small{Zero-shot Prompt Template (Shared by gpt-oss-120b \& GPT-5)}]
    \scriptsize
    \begin{verbatim}
You are an expert in regular expressions using the google-re2 engine.
Your task is to synthesize a regex that matches all positive examples and rejects all negative examples.
Output ONLY the regex pattern without any explanation or markdown formatting.

Technical Constraints:
1. Engine: Use google-re2 syntax.
2. Dot-All Mode: Assume 'dot nl' mode is enabled ('.' matches newline).
3. Whitespace: The shorthand '\s' must NOT match vertical tabs (\v).

Given the following examples:

[Positive Examples]
{positive_examples}

[Negative Examples]
{negative_examples}

Generate the google-re2 compatible regex pattern.
    \end{verbatim}
    \end{tcolorbox}
    \vspace{-0.2cm}
    \caption{The zero-shot prompt template used for both gpt-oss-120b and GPT-5 baselines.}
    \label{fig:prompt_template}
\end{figure*}
\section{Detailed Training Configurations}
\label{app:training_details}
We implemented the \resyn framework using PyTorch.
All models share a hidden size of $256$ with $8$ attention heads.
Regarding data coverage, \setregex and the \router were trained on the full dataset to learn global representations, whereas the \partitioner and \segmenter were trained exclusively on instances with top-level Union and Concatenation operators, respectively.

All experiments were conducted on a server equipped with four NVIDIA RTX A6000 GPUs.
The character and string encoders each consist of $2$ layers, whereas the \segmenter employs a $4$-layer encoder-decoder architecture.

We trained all modules using the AdamW optimizer with a batch size of $64$ and a learning rate of $5 \times 10^{-4}$.
The schedule included a linear warmup of $1,000$ steps followed by cosine annealing, with gradient clipping set to $1.0$.
Training proceeded for up to $100$ epochs, utilizing an early stopping strategy with a patience of $10$ epochs based on validation loss.

All modules were optimized using Negative Log-Likelihood (NLL) loss.
Specifically, the \router employed a class-weighted NLL ($w_c \propto 1 / N_c$) to mitigate label imbalance across decomposition strategies.

\end{document}